\documentclass[11pt]{article}
\usepackage{fullpage}
\usepackage{times,paralist}
\usepackage{amsmath,amsfonts,amssymb,amsthm}

\usepackage[colorlinks=true,linkcolor=red,citecolor=blue,urlcolor=red]{hyperref}
\usepackage{framed}
\usepackage{verbatim}

\begin{document}

\newtheorem{fact}{Fact}[section]
\newtheorem{rules}{Rule}[section]
\newtheorem{conjecture}{Conjecture}[section]
\newtheorem{theorem}{Theorem}[section]
\newtheorem{hypothesis}{Hypothesis}
\newtheorem{remark}{Remark}
\newtheorem{proposition}{Proposition}
\newtheorem{corollary}{Corollary}[section]
\newtheorem{lemma}{Lemma}[section]
\newtheorem{claim}{Claim}
\newtheorem{definition}{Definition}[section]

\newenvironment{proofof}[1]{\smallskip
\noindent {\bf Proof of #1.  }}{\hfill$\Box$
\smallskip}

\newenvironment{reminder}[1]{
\noindent {\bf Reminder of #1  }\em}{
}

\def \iouseful {\text{useful}}
\def \BP {{\sf BP}}
\def \coRP {{\sf coRP}}
\def \EXACT {{\sf EXACT}}
\def \SYM {{\sf SYM}}
\def \NOT {{\sf NOT}}
\def \SAC {{\sf SAC}}
\def \SUBEXP {{\sf SUBEXP}}
\def \ZPSUBEXP {{\sf ZPSUBEXP}}
\def \SYMACC {{\sf SYM\text{-}ACC} }
\def \QED {{\hfill$\Box$}}
\def \PH {{\sf PH}}
\def \RP {{\sf RP}}
\def \coMA {{\sf coMA}}
\def \ZPTISP {{\sf ZPTISP}}
\def \REXP {{\sf REXP}}
\def \coNP {{\sf coNP}}
\def \BPP {{\sf BPP}}
\def \NC {{\sf NC}}
\def \ZPE {{\sf ZPE}}
\def \NE {{\sf NE}}
\def \E {{\sf E}}
\def \poly { \text{\rm poly} }
\def \TC {{\sf TC}}
\def \DTS {{\sf DTS}}
\def \R {{\mathbb R}}
\def \Z {{\mathbb Z}}
\def \P {{\sf P}}
\def \MA {{\sf MA}}
\def \AM {{\sf AM}}
\def \MATIME {{\sf MATIME}}
\def \QP {{\sf QP}}
\def \coNQP {{\sf coNQP}}
\def \NP {{\sf NP}}
\def \EXP {{\sf EXP}}
\def \NTISP {{\sf NTISP}}
\def \DTISP {{\sf DTISP}}
\def \TISP {{\sf TISP}}
\def \T {{\sf TIME}}
\def \TIME {\T}
\def \Sig[#1] {{\sf \Sigma}_{#1} }
\def \Pie[#1] {{\sf \Pi}_{#1} }
\def \NTS {{\sf NTS}}
\def \NSP {{\sf NSPACE}}
\def \NSPACE {{\sf NSPACE}}
\def \BPACC {{\sf BPACC}}
\def \ACC {{\sf ACC}}
\def \ASP {{\sf ASPACE}}
\def \io {\textrm{\it io-}}
\def \ro {\textrm{\it ro-}}
\def \ATISP {{\sf ATISP}}
\def \SIZE {{\sf SIZE}}
\def \AT {{\sf ATIME}}
\def \AC {{\sf AC}}
\def \BPTIME {{\sf BPTIME}}
\def \SPACE {{\sf SPACE}}
\def \RE {{\sf RE}}
\def \ZPSUBEXP {{\sf ZPSUBEXP}}
\def \coREXP {{\sf coREXP}}
\def \NQL {{\sf NQL}}
\def \QL {{\sf QL}}
\def \RTIME {{\sf RTIME}}
\def \NSUBEXP {{\sf NSUBEXP}}
\def \MAEXP {{\sf MAEXP}}
\def \PP {{\sf PP}}
\def \PSPACE {{\sf PSPACE}}
\def \NT {{\sf NTIME}}

\def \NTIME {\NT}

\def \ATIME {\AT}

\def \Z {{\mathbb Z}}
\def \F {{\mathbb F}}
\def \NTIBI {{\sf NTIBI}}

\def \TIBI {{\sf TIBI}}

\def \QBFk {{\text{QBF}_k}}

\def \coNT {{\sf coNTIME}}

\def \coNTISP {{\sf coNTISP}}

\def \coNTIME {\coNT}

\def \coNTIBI {{\sf coNTIBI}}
\def \MOD {{\sf MOD}}
\def \ZPEXP {{\sf ZPEXP}}
\def \ZPTIME {{\sf ZPTIME}}
\def \ZPP  {{\sf ZPP}}
\def \DT {{\sf DTIME}}
\def \coNE {{\sf coNE}}
\def \DTIME {\DT}

\def \isin {\subseteq}

\def \isnotin {\nsubseteq}

\def \L {{\sf LOGSPACE}}

\def \LOGSPACE {\L}

\def \N {{\mathbb N}}
\def \NQP {{\sf NQP}}
\def \NEXP {{\sf NEXP}}

\def \coNEXP {{\sf coNEXP}}
\def \AND {{\sf AND}}
\def \OR {{\sf OR}}
\def \SIZE {{\sf SIZE}}
\def \sgn {\text{sgn}}
\def \eps {\varepsilon}

\newcommand{\card}[1]{\ensuremath{\left|#1\right|}}
\newcommand{\ip}[2]{\ensuremath{\left<#1,#2\right>}}
\newcommand{\mv}[2]{\ensuremath{\mathbf{MV}\!\left(#1,#2\right)}}
\newcommand{\ov}{\text{\bf OV}}

\title{The Orthogonal Vectors Conjecture\\ for Branching Programs and Formulas}
\author{
Daniel Kane\\UCSD \and Ryan Williams\footnote{Supported by an NSF CAREER.}\\MIT
}

\date{}

\maketitle

\begin{abstract} In the {\sc Orthogonal Vectors} (\ov{}) problem, we wish to determine if there is an orthogonal pair of vectors among $n$ Boolean vectors in $d$ dimensions. The \emph{OV Conjecture} (OVC) posits that \ov{} requires $n^{2-o(1)}$ time to solve, for all $d=\omega(\log n)$. 
Assuming the OVC, optimal time lower bounds have been proved for many prominent problems in $\P$, such as Edit Distance, Frechet Distance, Longest Common Subsequence, and approximating the diameter of a graph.

We prove that OVC is true in several computational models of interest: 

\begin{compactitem}
\item For all sufficiently large $n$ and $d$, \ov{} for $n$ vectors in $\{0,1\}^d$ has branching program complexity $\tilde{\Theta}(n\cdot \min(n,2^d))$. In particular, the lower bounds match the upper bounds up to polylog factors.
\item \ov{} has Boolean formula complexity $\tilde{\Theta}(n\cdot \min(n,2^d))$, over all complete bases of $O(1)$ fan-in. 
\item \ov{} requires $\tilde{\Theta}(n\cdot \min(n,2^d))$ wires, in formulas comprised of gates computing arbitrary symmetric functions of unbounded fan-in.
\end{compactitem}

Our lower bounds basically match the best known (quadratic) lower bounds for \emph{any} explicit function in those models. Analogous lower bounds hold for many related problems shown to be hard under OVC, such as Batch Partial Match, Batch Subset Queries, and Batch Hamming Nearest Neighbors, all of which have very succinct reductions to \ov{}.
 
The proofs use a certain kind of input restriction that is different from typical random restrictions where variables are assigned independently. We give a sense in which independent random restrictions cannot be used to show hardness, in that OVC is false in the ``average case'' even for $\AC^0$ formulas: 
\begin{compactitem}
\item For every fixed $p \in (0,1)$ there is an $\eps_p > 0$ such that for every $n$ and $d$, \ov{} instances where input bits are independently set to $1$ with probability $p$ (and $0$ otherwise) can be solved with $\AC^0$ formulas of size $O(n^{2-\eps_p})$, on all but a $o_n(1)$ fraction of instances. Moreover, $\eps_p \rightarrow 1$ as $p \rightarrow 1$.
\end{compactitem}

\end{abstract}

\thispagestyle{empty}
\newpage
\setcounter{page}{1}

\setlength{\parskip}{0.7ex plus 0.2ex minus 0.2ex}
\setlength{\parindent}{10pt}

\section{Introduction}

We investigate the following basic combinatorial problem:

\begin{framed} 
{\sc Orthogonal Vectors} (\ov{}) \\
Given: $n$ vectors $v_1,\ldots,v_n \in \{0,1\}^d$\\
Decide: Are there $i,j$ such that $\langle v_i,v_j\rangle = 0$?
\end{framed}

An instructive way of viewing the \ov{} problem is that we have a collection of $n$ sets over $[d]$, and wish to find two disjoint sets among them. The obvious algorithm runs in time $O(n^2 \cdot d)$, and $\log(n)$ factors can be shaved~\cite{Pritchard99}. For $d < \log_2(n)$, stronger improvements are possible: there are folklore $O(n \cdot 2^d \cdot d)$-time and $\tilde{O}(n + 2^d)$-time algorithms (for a reference, see~\cite{CSTExchange-Subset}). Truly subquadratic-time algorithms have recently been developed for even larger dimensionalities: the best known result in this direction is that for all constants $c \geq 1$, \ov{} with $d = c\log n$ dimensions can be solved in $n^{2-1/O(\log c)}$ time~\cite{AbboudWY15,DBLP:conf/soda/ChanW16}. However, it seems inherent that, as the vector dimension $d$ increases significantly beyond $\log n$, the time complexity of \ov{} approaches the trivial $n^2$ bound. 

Over the last several years, a significant body of work has been devoted to understanding the following plausible lower bound conjecture:

\begin{conjecture}[Orthogonal Vectors Conjecture (OVC) \cite{Williams05,DBLP:conf/icalp/AbboudVW14,BackursI15,AbboudBV15}]
For every $\eps > 0$, there is a $c \geq 1$ such that \ov{} cannot be solved in $n^{2-\eps}$ time on instances with $d = c\log n$. 
\end{conjecture}

In other words, OVC states that \ov{} requires $n^{2-o(1)}$ time on instances of dimension $\omega(\log n)$. The popular Strong Exponential Time Hypothesis~\cite{IP01,CIP09} (on the time complexity of CNF-SAT) implies OVC~\cite{Williams05}. For this reason, and the fact that the \ov{} problem is very simple to work with, the OVC has been the engine under the hood of many recent conditional lower bounds on classic problems solvable within $\P$. For example, the OVC implies nearly-quadratic time lower bounds for Edit Distance~\cite{BackursI15}, approximating the diameter of a graph~\cite{RoddityV13}, Frechet Distance~\cite{DBLP:conf/focs/Bringmann14,DBLP:journals/jocg/BringmannM16}, Longest Common Substring and Local Alignment~\cite{DBLP:conf/icalp/AbboudVW14}, Regular Expression Matching~\cite{DBLP:conf/focs/BackursI16}, Longest Common Subsquence, Dynamic Time Warping, and other string similarity measures~\cite{AbboudBV15,DBLP:conf/focs/BringmannK15}, Subtree Isomorphism and Largest Common Subtree~\cite{DBLP:conf/soda/AbboudBHVZ16}, Curve Simplification~\cite{buchin2016fine}, intersection emptiness of two finite automata~\cite{Wehar16}, first-order properties on sparse finite structures~\cite{DBLP:conf/soda/GaoIKW17} as well as average-case hardness for quadratic-time~\cite{DBLP:journals/iacr/BallRSV17}. Other works surrounding the OVC  (or assuming it) include~\cite{WilliamsY14,DBLP:conf/coco/Williams16,abboud2016approximation,DBLP:conf/lics/ChatterjeeDHL16,DBLP:conf/pods/AhlePR016,DBLP:journals/corr/EvaldD16,DBLP:conf/cpm/IliopoulosR16,DBLP:conf/soda/CairoGR16,KPS17}. 

Therefore it is of strong interest to prove the OVC in reasonable computational models. Note that \ov{} can be naturally expressed as a depth-three formula with unbounded fan-in: an $\OR$ of $n^2$ ${\sf NOR}$s of $d$ $\AND$s on two input variables: an $\AC^0$ formula of size $O(n^2 \cdot d)$. Are there smaller formulas for \ov{}?

\subsection{OVC is True in Restricted Models} 

In this paper, we study how well \ov{} can be solved in the Boolean formula and branching program models. Among the aforementioned \ov{} algorithms, only the first two seem to be efficiently implementable by formulas and branching programs: for example, there are DeMorgan formulas for \ov{} of size only $O(n^2 d)$ and size $O(n d 2^d)$, respectively (see Proposition~\ref{formula-ub}). 

The other algorithms do not seem to be implementable in small space, in particular with small-size branching programs. Our first theorem shows that the simple constructions solving \ov{} with $O(n^2 \cdot d)$ and $O(n \cdot 2^d \cdot d)$ work are essentially optimal for all choices of $d$ and $n$:

\begin{theorem}[OVC For Formulas of Bounded Fan-in] \label{formula-lb} For every constant $c \geq 1$, \ov{} on $n$ vectors in $d$ dimensions does not have $c$-fan-in formulas of size $O(\min\{n^2/(\log d),n \cdot 2^d/(d^{1/2}\log d)\})$, for all sufficiently large $n,d$.
\end{theorem}

\begin{theorem}[OVC For Branching Programs] \label{bp-lb} \ov{} on $n$ vectors in $d$ dimensions does not have branching programs of size $O(\min\{n^2,n \cdot 2^d/(d^{1/2})\}/(\log(nd)\log(d)))$, for all sufficiently large $n,d$.
\end{theorem}

As far as we know, size-$s$ formulas of constant fan-in may be more powerful than size-$s$ branching programs (but note that DeMorgan formulas can be efficiently simulated with branching programs). Thus the two lower bounds are incomparable. These lower bounds are tight up to the (negligible) factor of $\min\{\sqrt{\log n},d^{1/2}\}\log(d)\log(nd)$, as the following simple construction shows: 

\begin{proposition}\label{formula-ub} \ov{} has $\AC^0$ formulas (and branching programs) of size $O(dn \cdot \min(n,2^d))$.
\end{proposition}

\begin{proof}
The $O(dn^2)$ bound is obvious: take an OR over all $\binom{n}{2}$ pairs of vectors, and use an $\AND \circ \OR$ of $O(d)$ size to determine orthogonality of the pair. For the $O(dn 2^d)$ bound, our strategy is to try all $2^d$ vectors $v$, and look for a $v$ that is equal to one input vector and is orthogonal to another input vector. To this end, take an OR over all $2^d$ possible vectors $w$ over $[d]$, and take the AND of two conditions:
\begin{compactenum}
\item There is a vector $v$ in the input such that $v=w$. This can be computed with an an OR over all $n$ vectors of an $O(d)$-size formula, in $O(n d)$ size.
\item There is a vector $u$ in the input such that $\langle u,w\rangle = 0$. This can be computed with a \emph{parallel} OR over all $n$ vectors of an $O(d)$-size formula, in $O(n d)$ size.
\end{compactenum} Note that the above formulas have constant-depth, with unbounded fan-in AND and OR gates. Since DeMorgan formulas of size $s$ can be simulated by branching programs of size $O(s)$, the proof is complete.\footnote{This should be folklore, but we couldn't find a reference; see the Appendix~\ref{deMorgan-bp}.}
\end{proof}

\paragraph{Formulas with symmetric gates.} As mentioned above, \ov{} can be naturally expressed as a depth-three formula of unbounded fan-in: an $\AC^0_3$ formula of $O(n^2 d)$ wires. We show that this wire bound is also nearly optimal, even when we allow arbitrary symmetric Boolean functions as gates. Note this circuit model subsumes both $\AC$ (made up of \AND, \OR, and \NOT{} gates) and $\TC$ (made up of {\sf MAJORITY} and {\sf NOT} gates).

\begin{theorem} \label{symmetric-lb} Every formula computing \ov{} composed of arbitrary symmetric functions with unbounded fan-in needs at least $\Omega(\min\{n^2/(\log d),n\cdot 2^d/(d^{1/2}\log d)\})$ wires, for all $n$ and $d$.
\end{theorem}

\subsection{Lower Bounds for Batch Partial Match, Batch Subset Query, Batch Hamming Nearest Neighbors, etc.} A primary reason for studying \ov{} is its ubiquity as a ``bottleneck'' special case of many other basic search problems. In particular, many problems have very succinct reductions from \ov{} to them, and our lower bounds extend to these problems. 

We say that a \emph{linear projection reduction} from a problem $A$ to problem $B$ is a circuit family $\{C_n\}$ where each $C_n$ has $n$ input and $O(n)$ outputs, each output of $C_n$ depends on at most one input, and $x \in A$ if and only if $C_{|x|}(x) \in B$, for all possible inputs $x$. Under this constrained reduction notion, it is easy to see that if $\ov{}$ has a linear projection reduction to $B$, then size lower bounds for $\ov{}$ (even in our restricted settings) imply analogous lower bounds for $B$ as well. Via simple linear projection reductions which preserve both $n$ and $d$ (up to constant multiplicative factors), analogous lower bounds hold for many other problems which have been commonly studied, such as:

\begin{framed} 
{\sc Batch Partial Match} \\
Given: $n$ ``database'' vectors $v_1,\ldots,v_n \in \{0,1\}^d$ and $n$ queries $q_1,\ldots,q_n \in \{0,1,\star\}^d$ \\
Decide: Are there $i,j$ such that $v_i$ is a partial match of $q_j$, i.e. for all $k$, $q_j[k] \in \{v_i[k],\star\}$?
\end{framed}

\begin{framed} 
{\sc Batch Subset Query} \\
Given: $n$ sets $S_1,\ldots,S_n \subseteq [d]$ and $n$ queries $T_1,\ldots,T_n \subseteq [d]$ \\
Decide: Are there $i,j$ such that $S_i \subseteq T_j$?
\end{framed}

\begin{framed} 
{\sc Batch Hamming Nearest Neighbors} \\
Given: $n$ points $p_1,\ldots,p_n \in \{0,1\}^d$ and $n$ queries $q_1,\ldots,q_n \in \{0,1\}^d$, integer $k$\\
Decide: Are there $i,j$ such that $p_i$ and $q_j$ differ in at most $k$ positions?
\end{framed}
 
\subsection{``Average-Case'' OVC is False, Even for AC0}

The method of proof in the above lower bounds is an input restriction method that does \emph{not} assign variables independently (to $0$, $1$, or $\star$) at random. (Our restriction method could be viewed as a random process, just not one that assigns variables independently.) Does \ov{} become easier under natural product distributions of instances, e.g., with each bit of each vector being an independent random variable? Somewhat surprisingly, we show that a reasonable parameterization of average-case OVC is false, even for $\AC^0$ formulas. 

For $p \in (0,1)$, and for a given $n$ and $d$, we call $\ov{}(p)_{n,d}$ the distribution of $\ov{}$ instances where all bits of the $n$ vectors are chosen independently, set to $1$ with probability $p$ and $0$ otherwise. We would like to understand when $\ov{}(p)$ can be efficiently solved on almost all instances (i.e., with probability $1-o(1)$). We give formulas of truly sub-quadratic size for every $p > 0$: 

\begin{theorem}\label{avg-case-bp} For every $p \in (0,1)$, and every $n$ and $d$, there is an $\AC^0$ formula of size $n^{2-\eps_p}$ that correctly answers all but a $o_n(1)$ fraction of $\ov{}(p)_{n,d}$ instances on $n$ vectors and $d$ dimensions, for an $\eps_p > 0$ such that $\eps_p \rightarrow 1$ as $p \rightarrow 1$.
\end{theorem}

Interestingly, our $\AC^0$ formulas have one-sided error, even in the worst case: if there is no orthogonal pair in the instance, our formulas \emph{always} output $0$. However, they may falsely report that there is no orthogonal pair, but this only occurs with probability $o(1)$ on a random $\ov{}(p)_{n,d}$ instance, for any $n$ and $d$.

\subsection{Intuition}

Our lower bounds give some insight into what makes \ov{} hard to solve. There are two main ideas:
\begin{compactenum}
\item \ov{} instances with $n$ $d$-dimensional vectors can encode difficult Boolean functions on $d$ inuts, requiring circuits of size $\tilde{\Omega}(\min(2^d,n))$. This can be accomplished by encoding those strings with ``middle'' Hamming weight from the truth table of a hard function with the vectors in an \ov{} instance, in such a way that finding an orthogonal pair is equivalent to evaluating the hard Boolean function at a given $d$-bit input. This is an inherent property of \ov{} that is independent of the computational model.
\item
Because we are working with simple computational models, we can generally make the following kind of claim: given an algorithm for solving \ov{} and given a partial assignment to all input vectors except for one appropriately chosen vector, we can propagate this partial assignment through the algorithm, and ``shrink'' the size of the algorithm by a factor of $\Omega(n)$. This sort of argument was first used by Nechiporuk~\cite{N66} in the context of branching program lower bounds, and can be also applied to formulas.
\end{compactenum}

Combining the two ideas, if we can ``shrink'' our algorithm by a factor of $n$ by restricting the inputs appropriately, and argue that the remaining subfunction requires circuits of size $\tilde{\Omega}(\min(2^d,n))$, we can conclude that the original algorithm for \ov{} must have had size $\tilde{\Omega}(\min(n 2^d,n^2))$. (Of course, there are many details to verify, but this is the basic idea.)

The small $\AC^0$ formulas for $\ov{}(p)$ (the average-case setting) involve several ideas. First, given the probability $p \in (0,1)$ of $1$ and the number of vectors $n$, we observe a simple phase transition phenomenon: there is only a particular range of dimensionality $d$ in which the problem is non-trivial, and outside of this range, almost all instances are either ``yes'' instances or ``no'' instances. Second, within this ``hard'' range of $d$, the orthogonal vector pairs are expected to have a special property: with high probability, at least one orthogonal pair in a ``yes'' instance has noticeably fewer ones than a typical vector in the distribution. To obtain a sub-quadratic size $\AC^0$ formula from these observations, we partition the instance into small groups such that the orthogonal pair (if it exists) is the only ``sparse'' vector in its group, whp. Over all pairs of groups $i,j$ in parallel, we take the component-wise OR of all sparse vectors in group $i$, and similarly for group $j$. Then we test the two ORed vectors for orthgonality. By doing so, if our formula ever reports $1$, then there is some orthogonal pair in the instance (even in the worst case).

\section{Lower Bounds}

\paragraph{Functions that are hard on the middle layer of the hypercube.} In our lower bound proofs, we will use functions on $d$-inputs for which every small circuit fails to agree with the function on inputs of Hamming weight about $d/2$.  Let $\binom{[d]}{k}$ denote the set of all $d$-bit vectors of Hamming weight $k$. 

\begin{lemma}\label{hard-inputs} Let $d$ be even, let ${\cal C}$ be a set of Boolean functions, let $N(d,s)$ be the number of functions in ${\cal C}$ on $d$ inputs of size at most $s$, and let $s^{\star} \in \N$ satisfy $\log_2(N(d,s^{\star})) < \binom{d}{d/2}$. 

Then there is a sequence of $\binom{d}{d/2}$ pairs $(x_i,y_i) \in \binom{[d]}{d/2} \times \{0,1\}$, such that every function $f : \{0,1\}^d \rightarrow \{0,1\}$ satisfying $f(x_i) = y_i$ (for all $i=1,\ldots,\binom{d}{d/2}$) requires ${\cal C}$-size at least $s^{\star}$.
\end{lemma}

\begin{proof} By definition, there are $N(d,s)$ functions of size $s$ on $d$ inputs from ${\cal C}$, and there are $2^{\binom{d}{d/2}}$ input/output sequences $(x_i,y_i)\binom{[d]}{d/2} \times \{0,1\}$ defined over all $d$-bit vectors of Hamming weight $d/2$. For $2^{\binom{d}{d/2}} > N(d,s)$, there is an input/output sequence that is not satisfied by any function in ${\cal C}$ of size $s$. 
\end{proof}

Note that it does not matter \emph{what} is meant by size in the above lemma: it could be gates, wires, etc., and the lemma still holds (as it is just counting). The above simple lemma applies to formulas, as follows:

\begin{corollary}\label{hard-middle-formula} Let $c \geq 2$ be a constant. There are $\binom{d}{d/2}$ pairs $(x_i,y_i) \in \binom{[d]}{d/2} \times \{0,1\}$, such that every function $f : \{0,1\}^d \rightarrow \{0,1\}$ satisfying $f(x_i) = y_i$ (for all $i=1,\ldots,\binom{d}{d/2}$) needs $c$-fan-in formulas of size at least $\Omega(2^d/(d^{1/2}\log d))$.
\end{corollary}

\begin{proof}
There are $N(d,s) \leq d^{k_c \cdot s}$ formulas of size $s$ on $d$ inputs, where the constant $k_c$ depends only on $c$. When $2^{\binom{d}{d/2}} > d^{k_c \cdot s}$, Lemma~\ref{hard-inputs} says that there is an input/output sequence of length $\binom{d}{d/2}$ that \emph{no} formula of size $s$ can satisfy. Thus to satisfy that sequence, we need a formula of size $s$ at least large enough that $2^{\binom{d}{d/2}} \leq d^{k_c \cdot s}$, i.e., $s \geq\Omega\left(\binom{d}{d/2}/\log(d)\right) \geq \Omega(2^d/(d^{1/2}\log d))$.
\end{proof}

\subsection{Lower Bound for Constant Fan-in Formulas} We are now ready to prove the lower bound for Boolean formulas of constant fan-in:

\begin{reminder}{Theorem~\ref{formula-lb}} For every constant $c \geq 1$, \ov{} on $n$ vectors in $d$ dimensions does not have $c$-fan-in formulas of size $O(\min\{n^2/(\log d),n \cdot 2^d/(d^{1/2}\log d)\})$, for all sufficiently large $n,d$.
\end{reminder}

All of the lower bound proofs have a similar structure. We will give considerably more detail in the proof of Theorem~\ref{formula-lb} to aid the exposition of the later lower bounds.

\begin{proof} To simplify the calculations, assume $d$ is even in the following. Let $F_{n,d}(v_1,\ldots,v_n)$ be a $c$-fan-in formula of minimal size $s$ computing $\ov{}$ on $n$ vectors of dimension $d$, where each $v_i$ denotes a sequence of $d$ Boolean variables $(v_{i,1},\ldots,v_{i,d})$. 

Let $\ell$ be the number of leaves of $F_{n,d}$. Since $F_{n,d}$ is minimal, each gate has fan-in at least two (gates of fan-in $1$ can be ``merged'' into adjacent gates). Therefore (by an easy induction on $s$) we have
\begin{align}\label{leaves-size}
s \geq \ell \geq s/2.
\end{align}
Observe there must be a vector $v_{i^{\star}}$ (for some $i^{\star} \in [n]$) whose $d$ Boolean variables appear on at most $\ell/n$ leaves of the formula $F_{n,d}$.  

{\bf Case 1.} Suppose $\binom{d}{d/2} \leq n-1$. Let $\{(x_i,y_i)\} \subseteq \binom{[d]}{d/2} \times \{0,1\}$ be a set of hard pairs from Corollary~\ref{hard-middle-formula}, and let $f : \{0,1\}^d \rightarrow \{0,1\}$ be any function that satisfies $f(x_i) = y_i$, for all $i$. Let $\{x'_1,\ldots,x'_t\} \subseteq \binom{[d]}{d/2}$ be those $d$-bit strings of Hamming weight $d/2$ such that $f(x'_i) = 1$, for some $t \leq \binom{d}{d/2} \leq n-1$. By Corollary~\ref{hard-middle-formula}, such an $f$ needs $c$-fan-in formulas of size at least $\Omega(2^d/(d^{1/2}\log d))$.

{\bf Case 2.} Suppose $\binom{d}{d/2} \geq n-1$. Then we claim there is a list of input/output pairs $(x_1,y_1)$, $\ldots$, $(x_{n-1},y_{n-1}) \in \binom{[d]}{d/2} \times \{0,1\}$ such that for every $f : \{0,1\}^d \rightarrow \{0,1\}$ satisfying $f(x_i) = y_i$, for all $i$, $f$ needs formulas of size at least $\Omega(n/\log d)$. To see this, simply note that if we take $n-1$ distinct strings $x_1,\ldots,x_{n-1}$ from $\binom{d}{d/2}$, there are $2^{n-1}$ possible choices for the list of pairs. So when $2^{n-1} > d^{k_c \cdot s}$, there is a list $(x_1,y_1)$, $\ldots$, $(x_{n-1},y_{n-1})$ that no formula of size $s$ satisfies. For any function $f : \{0,1\}^d \rightarrow \{0,1\}$ such that $f(x_i)=y_i$ for all $i=1,\ldots,n-1$, its formula size $s \geq \Omega(n/\log d)$ in this case. Let $\{x'_1,\ldots,x'_t\} \subseteq \binom{[d]}{d/2}$ be those $d$-bit strings of Hamming weight $d/2$ such that $(x'_i,1)$ is on the list, for some $t \leq n-1$. 

For either of the two cases, we will use the list of $t \leq \min\{n-1,\binom{d}{d/2}\}$ strings $\{x'_1,\ldots,x'_t\}$ to make assignments to the variables $v_i$ of our \ov{} formula, for all $i \neq i^{\star}$. In particular, for all $i=1,\ldots,t$ with $i \neq i^{\star}$, we substitute the $d$ bits of $\overline{x'_i}$ (the complement of $x_i$, obtained by flipping all the bits of $x'_i$) in place of the $d$-bit input vector $v_i$. If $t < n-1$ (which can happen in case 1), substitute all other $\vec{z_j}$ with $j \neq i^{\star}$ with $\vec{1}$. Note that all of the pairs of vectors substituted so far are not orthogonal to each other: for all $i \neq i'$, we have $\langle x'_i,x'_{i'}\rangle \neq 0$, because both $x'_i$ and $x'_{i'}$ are distinct vectors each with $d/2$ ones, and for all $i$ we have $\langle x'_i,\vec{1}\rangle \neq 0$.

After these substitutions, the remaining formula $F'_n$ is on only $d$ inputs, namely the vector $v_{i^{\star}}$. Moreover, $F'_n$ is a formula with at most $\ell/n$ leaves labeled by literals: the rest of the leaves are labeled with 0/1 constants. After simplifying the formula (replacing all gates with some 0/1 inputs by equivalent functions of smaller fan-in, and replacing gates of fan-in $1$ by wires), the total number of leaves of $F'_n$ is now at most $\ell/n$. Therefore by \eqref{leaves-size} we infer that 
\begin{align}\label{ub-Fprime}
\text{size}(F'_n) &\leq 2\ell/n.
\end{align} 
Since $F_{n,d}$ computes \ov{}, it follows that for every input vector $y \in \{0,1\}^d$ of Hamming weight $d/2$, $F'_n$ on input $y$ outputs $1$ if and only if there is some $i$ such that $\langle \overline{x'_i},y \rangle = 0$. Note that since both $\overline{x'_i}$ and $y$ have Hamming weight exactly $d/2$, we have $\langle \overline{x_i},y\rangle = 0$ if and only if $y = x_i$. 

By our choice of $x_i$'s, it follows that for all $y \in \{0,1\}^d$ of Hamming weight $d/2$, $F'_n(y) = 1$ if and only if $f(y)=1$. By our choice of $f$ (from Corollary~\ref{hard-middle-formula} in case 1, and our claim in case 2), we must have
\begin{align}\label{lb-Fprime}
\text{size}(F'_n) &\geq \min\{\Omega(2^d/(d^{1/2}\log d)),\Omega(n/\log d))\},
\end{align} 
depending on whether $\binom{d}{d/2} \leq n-1$ or not (case 1 or case 2). Combining \eqref{ub-Fprime} and \eqref{lb-Fprime}, we infer that
\begin{align}
\ell &\geq \Omega(n \cdot \min\{\Omega(2^d/(d^{1/2}\log d)),\Omega(n/\log d))\}),
\end{align} 
therefore the overall lower bound on formula size is $s \geq \Omega\left(\min\left\{\frac{n^2}{\log d},\frac{n \cdot 2^d}{d^{1/2}\log d}\right\}\right)$.
\end{proof}

\paragraph{Remark on a Red-Blue Variant of OV.} In the literature, \ov{} is sometimes posed in a different form, where half of the vectors are colored red, half are colored blue, and we wish to find a red-blue pair which is orthogonal. Calling this form \ov{}', we note that \ov{}' also exhibits the same lower bound up to constant factors. Given an algorithm/formula/circuit $A$ for computing \ov{}' on $2n$ vectors ($n$ of which are red, and $n$ of which are blue), it is easy to verify that an algorithm/formula/circuit for \ov{} on $n$ vectors results by simply putting two copies of the set of vectors in the red and blue parts. Thus our lower bounds hold for the red-blue variant as well.

\subsection{Lower Bound for Branching Programs} 

Recall that a branching program of size $S$ on $n$ variables is a directed acyclic graph $G$ on $S$ nodes, with a distinguished start node $s$ and exactly two sink nodes, labeled $0$ and $1$ respectively. All non-sink nodes are labeled with a variable $x_i$ from $\{x_1,\ldots,x_n\}$, and have one outgoing edge labeled $x_i = 1$ and another outgoing edge labeled $x_i = 0$. The branching program $G$ \emph{evaluated at an input} $(a_1,\ldots,a_n) \in \{0,1\}^n$ is the subgraph obtained by only including edges of the form $x_i = a_i$, for all $i=1,\ldots,n$. Note that after such an evaluation, the remaining subgraph has a unique path from the start node $s$ to a sink; the sink reached on this unique path (be it $0$ or $1$) is defined to be the output of $G$ on $(a_1,\ldots,a_n)$. 


\begin{reminder}{Theorem~\ref{bp-lb}} \ov{} on $n$ vectors in $d$ dimensions does not have branching programs of size $O(\min\{n^2,n \cdot 2^d/(d^{1/2})\}/(\log(nd)\log(d)))$, for all sufficiently large $n,d$.
\end{reminder}

\begin{proof} (Sketch) The proof is similar to Theorem~\ref{formula-lb}; here we focus on the steps of the proof that are different. Let $G$ be a branching program with $S$ nodes computing \ov{} on $n$ vectors with $d$ dimensions. Each node of $G$ reads a single input bit from one of the input vectors; thus there is an input vector $v_{i^{\star}}$ that is read only $O(S/n)$ times in the entire branching program $G$. 

We will assign all variables other than the $d$ variables that are part of$v_{i^{\star}}$. Using the same encoding as Theorem~\ref{formula-lb}, by assigning the $n-1$ other vectors, we can implement a function $f : \{0,1\}^d \rightarrow \{0,1\}$ that is hard for branching programs to compute on the set of $d$-bit inputs in $\binom{[d]}{d/2}$. In particular, we substitute $d$-bit vectors which represent inputs from $f^{-1}(1) \cap \binom{[d]}{d/2}$ for all $n-1$ input vectors different from $v_{i^{\star}}$. For each of these assignments, we can reduce the size of the branching program accordingly: for each input bit $x_j$ that is substituted with the bit $a_j$, we remove all edges with the label $x_j = \neg a_j$, so that every node labeled $x_j$ now has outdegree $1$. After the substitution, two properties hold:

\begin{enumerate}
\item There is a hard function $f$ such that the minimum size $T$ of a branching program computing $f$ on the $n-1$ inputs satisfies $T \log_2(T) \geq \Omega(\min\{\binom{d}{d/2},n\}/\log(d))$. To see this is possible, first note there are $d^T \cdot 2^{\Theta(T \log(T))}$ branching programs of size $T$ on $d$ inputs (there are $d^T$ choices for the node labels, and $2^{\Theta(T \log(T))}$ choices for the remaining graph on $T$ nodes). In contrast, there are at least $2^{\min\{\binom{d}{d/2},n-1\}}$ choices for the hard function $f$'s values on $d$-bit inputs of Hamming weight $d/2$. Therefore there is a function $f$ such that $d^T \cdot 2^{\Theta(T \log(T))} \geq 2^{\min\{\binom{d}{d/2},n-1\}}$, or 
\[T + \Theta(T \log(T)) \geq \min\left\{\binom{d}{d/2},n-1\right\}/\log_2(d).\]

\item The minimum size of a branching program computing a function $f:\{0,1\}^d \rightarrow\{0,1\}$ on the remaining $d$ bits of input is at most $O(S/n)$. This follows because every node $v$ with outdegree $1$ can be removed from the branching program without changing its functionality: for every arc $(u,v)$ in the graph, we can replace it with the arc $(u,v')$, where $(v,v')$ is the single edge out of $v$, removing the node $v$.
\end{enumerate}
Combining these two points, we have $(S/n) \cdot \log(S/n) \geq \Omega\left(\min\left\{\binom{d}{d/2},n\right\}/\log(d)\right)$, or
\[S \geq \Omega\left(\frac{\min\{n\binom{d}{d/2},n^2\}}{\log(S/n)\cdot\log(d)}\right).\] Since $S \leq n^2 d$, we have \[S \geq \Omega\left(\frac{\min\{n\binom{d}{d/2},n^2\}}{\log(nd)\cdot\log(d)}\right) \geq \Omega\left(\frac{\min\{n\cdot 2^d/d^{1/2},n^2\}}{\log(nd)\cdot\log(d)}\right).\]
This concludes the proof.
\end{proof}

\subsection{Formulas With Symmetric Gates}

We will utilize a lower bound on the number of functions computable by symmetric-gate formulas with a small number of wires: 

\begin{lemma}\label{sym-size}
There are $n^{O(w)}$ symmetric-gate formulas with $w$ wires and $n$ inputs.
\end{lemma}

\begin{proof} There is an injective mapping from the set of trees of unbounded fan-in and $w$ wires into the set of binary trees with at most $2w$ nodes: simply replace each node of fan-in $k$ with a binary tree of at most $2k$ nodes. The number of such binary trees is $O(4^{2w})$ (by upper bounds on Catalan numbers). This counts the number of ``shapes'' for the symmetric formula; we also need to count the possible gate assignments. There are $2^{k+1}$ symmetric functions on $k$ inputs. So for a symmetric-gate formula with $g$ gates, where the $i$th gate has fan-in $w_i$ for $i=1,\ldots,g$, the number of possible assignments of symmetric functions to its gates is $\prod_{i=1}^g 2^{w_i+1} = 2^{g+\sum_i w_i} = 2^{g+w}$. There are at most $w$ leaves, and there are $n^w$ ways to choose the variables read at each leaf. Since $g \leq w$, we conclude that there are at most $4^{2w}\cdot 2^{2w} \cdot n^w \leq n^{O(w)}$ symmetric-gate formulas with $w$ wires.
\end{proof}

\begin{reminder}{Theorem~\ref{symmetric-lb}} Every formula computing \ov{} composed of arbitrary symmetric functions with unbounded fan-in needs at least $\Omega(\min\{n^2/(\log d),n\cdot 2^d/(d^{1/2}\log d))\})$ wires, for all $n$ and $d$.
\end{reminder}
 
\begin{proof} (Sketch) The proof is quite similar to the other lower bounds, given Lemma~\ref{sym-size}, so we just sketch the ideas. Let $F$ be a symmetric-gate formula for computing $\ov{}$ with unbounded fan-in and $w$ wires. Let $w_i$ be the number of wires touching inputs and $w_g$ be the number of wires that do not touch inputs. Since $F$ is a formula, we have (by a simple induction argument) that $w_i \geq w_g$, thus\begin{align}\label{wires}w \leq 2w_i.\end{align} 
As before, each leaf of the formula is labeled by an input from one of the input $n$ vectors; in this way, every leaf is ``owned'' by one of the $n$ input vectors. We will substitute a 0/1 variable assignment to all vectors, except the vector $\vec{z^{\star}}$ which owns the fewest leaves. This gives a 0/1 assignment to all but $O(w_i/n)$ of the $w_i$ wires that touch inputs. 

After any such variable assignment, we can simplify $F$ as follows: for every symmetric-function gate $g$ which has $w_g$ input wires with $k$ wires assigned 0/1, we can replace $g$ with a symmetric function $g'$ that has only $w_g-k$ inputs, and no input wires assigned 0/1 (a partial assignment to a symmetric function just yields another symmetric function on a smaller set of inputs). If $g'$ is equivalent to a constant function itself, then we remove it from the formula and substitute its output wire with that constant, repeating the process on the gates that use the output of $g$ as input. When this process completes, our new formula $F'$ has $d$ inputs and no wires that are assigned constants. So $F'$ has $O(w_i/n)$ wires touching inputs, and therefore by \eqref{wires} the total number of wires in $F'$ is $O(w/n)$. 

As described earlier, the $n-1$ vectors we assign can implement $2^{\min\{n-1,\binom{d}{d/2}\}}$ different functions on $d$-bit inputs, but there are at most $d^{O(w/n)}$ functions computable by the symmetric formula remaining, by Lemma~\ref{sym-size}. Thus we need that the number of wires $w$ satisfies $d^{O(w/n)}\geq 2^{\min\{n-1,\binom{d}{d/2}\}}$, or \[w \geq \Omega(\min\{n^2, n \cdot 2^d/(d^{1/2})\}/(\log d)).\] This completes the proof.
\end{proof}

\section{Small Formulas for OV in the Average Case}\label{sec:avg-case}

Recall that for $p \in (0,1)$ and for a fixed $n$ and $d$, we say that $\ov{}(p)_{n,d}$ is the distribution of $\ov{}$ instances where all bits of the $n$ vectors from $\{0,1\}^d$ are chosen independently, set to $1$ with probability $p$ and $0$ otherwise. We will often say that a vector is ``sampled from $\ov{}(p)$'' if each of its bits are chosen independently in this way. We would like to understand how efficiently $\ov{}(p)_{n,d}$ can be solved on almost all instances (i.e., with probability $1-o(1)$), for every $n$ and $d$. 

\begin{reminder}{Theorem~\ref{avg-case-bp}} For every $p \in (0,1)$, and every $n$ and $d$, there is an $\AC^0$ formula of size $n^{2-\eps_p}$ that correctly answers all but a $o_n(1)$ fraction of $\ov{}(p)_{n,d}$ instances on $n$ vectors and $d$ dimensions, for an $\eps_p > 0$ such that $\eps_p \rightarrow 1$ as $p \rightarrow 1$.
\end{reminder}

\begin{proof} Let $\eps > 0$ be sufficiently small in the following. 
First, we observe that $\ov{}(p)_{n,d}$ is very easy, unless $d$ is close to $(2/\log_2(1/(1-p^2)))\log_2(n)$. In particular, for dimensionality $d$ that is significantly smaller (or larger, respectively) than this quantity, all but a $o(1)$ fraction of the $\ov{}(p)_{n,d}$ instances are ``yes'' (or ``no'', respectively). To see this, note that two randomly chosen $d$-dimensional vectors under the $\ov{}(p)_{n,d}$ distribution are orthogonal with probability $(1-p^2)^d$. For $d = (2/\log_2(1/(1-p^2)))\log_2(n)$, so a random pair is orthogonal with probability \[(1-p^2)^{(2/\log_2(1/(1-p^2)))\log_2(n)} = 1/n^2.\] Thus an $\ov{}(p)_{n,d}$ instance with $n$ vectors has nontrivial probability of being a yes instance for $d$ approximately $(2/\log_2(1/(1-p^2)))\log_2(n)$.

Therefore if $d > (2/\log_2(1/(1-p^2))+\eps)\log_2(n)$, or  $d < (2/\log_2(1/(1-p^2))-\eps)\log_2(n)$, then the random instance is either almost surely a ``yes'' instance, or almost surely a ``no'' instance, respectively. These comparisons could be done with the quantities $(2/\log_2(1/(1-p^2))-\eps)\log_2(n)$ and $(2/\log_2(1/(1-p^2))+\eps)\log_2(n)$ (which can be hard-coded in the input) with a $\poly(d,\log n)$-size branching program, which can output $0$ and $1$ respectively if this is the case.\footnote{As usual, $\poly(m)$ refers to an unspecified polynomial of $m$ of fixed degree.}

From here on, assume that $d \in [(2/\log_2(1/(1-p^2))-\eps)\log_2(n),(2/\log_2(1/(1-p^2))+\eps)\log_2(n)]$. Note that for $p$ sufficiently close to $1$, the dimensionality $d$ is $\delta \log n$ for a small constant $\delta > 0$ that is approaching $0$. Thus in the case of large $p$, the $\AC^0$ formula given in Proposition~\ref{formula-ub} has sub-quadratic size. In particular, the size is
\begin{align} \label{first-subquad} O(n \cdot 2^d \cdot d) \leq n^{1+2/\log_2(1/(1-p^2))+o(1)}.\end{align} For $p \geq 0.867 > \sqrt{3/4}$, this bound is sub-quadratic. For smaller $p$, we will need a more complex argument.

Suppose $u,v \in \{0,1\}^d$ are randomly chosen according to the distribution of $\ov{}(p)$ (we will drop the $n,d$ subscript, as we have fixed $n$ and $d$ at this point). 


We now claim that, conditioned on the event that $u,v$ is an orthogonal pair, both $u$ and $v$ are expected to have between $(p/(1+p)-\eps)d$ and $(p/(1+p)+\eps)d$ ones, with $1-o(1)$ probability. The event that both $u[i]=v[i]=1$ holds with probability $1-p^2$; conditioned on this event never occurring, we have 
\begin{align*}
\Pr[u[i] = 0, v[i] = 0 ~\mid~ \neg(u[i]=v[i]=1)] &= (1-p)^2/(1-p^2),\\
\Pr[u[i] = 1, v[i] = 0 ~\mid~ \neg(u[i]=v[i]=1)] &= p(1-p)/(1-p^2),\\
\Pr[u[i] = 0, v[i] = 1 ~\mid~ \neg(u[i]=v[i]=1)] &= p(1-p)/(1-p^2).
\end{align*}
Hence the expected number of ones in $u$ (and in $v$) is only $p(1-p)d/(1-p^2) = 
pd/(1+p)$, and the number of ones is within $(-\eps d,\eps d)$ of this quantity with probability $1-o(1)$. (For example, in the case of $p=1/2$, the expected number of ones is $d/3$, while a typical vector has $d/2$ ones.)

Say that a vector $u$ is \emph{light} if it has at most $(p/(1+p)+\eps)d$ ones. It follows from the above discussion that if an $\ov{}(p)$ instance is a ``yes'' instance, then there is an orthogonal pair with two light vectors, with probability $1-o(1)$. Since the expected number of ones is $pd$, the probability that a randomly chosen $u$ is light is
\begin{align*}\label{light-bd}\nonumber
\Pr\left[u \text{ has at most $\left(\frac{p}{1+p}+\eps\right)d=pd\left(1-\frac{p}{p+1}+\frac{\eps}{p}\right)$ ones }\right] &\leq e^{-(p/(p+1)+\eps/p)^2 pd/2}\\ &= e^{-(p^3/(2(p+1)^2)-\Theta_p(\eps))d},
\end{align*}
by a standard Chernoff tail bound (see Theorem~\ref{chernoff-bd} in Appendix~\ref{chernoff}). So with high probability, there are at most $n \cdot e^{-(p^3/(2(p+1)^2)-\Theta_p(\eps))d} = n^{1-\alpha}$ light vectors in an $\ov{}(p)$ instance, where \[\alpha = \frac{p^3}{(p+1)^2\log_2(1/(1-p^2))}+\Theta_p(\eps)\cdot 2/(\log_2(1/(1-p^2))).\] 

Divide the $n$ vectors of the input arbitrarily into $n^{1-\alpha(1-\eps)}$ groups $G_1,\ldots,G_{n^{1-\alpha(1-\eps)}}$, of $O(n^{\alpha(1-\eps)})$ vectors each. WLOG, suppose an orthogonal pair $u,v$ lies in different groups $u \in G_i$ and $v \in G_j$, with $i \neq j$ (note that, conditioned on there being an orthogonal pair, this event also occurs with $1-o(1)$ probability). Since every vector is independently chosen, and given that $\Pr_v[v \text{ is light }] \leq 1/n^{\alpha}$, note that 
\[\Pr_{v_1,\ldots,v_{n^{\alpha(1-\eps)}}}[\text{all $v_i$ in group $G_a$ are not light}] \geq (1-1/n^{\alpha})^{n^{\alpha(1-\eps)}} \geq 1-1/n^{\eps\alpha},\] for every group $G_a$. Thus the groups $G_i$ and $G_j$ have at most one light vector with probability $1-o(1)$.

We can now describe our formula for $\ov{}(p)$, in words. Let \emph{Light}$(v)$ be the function which outputs $1$ if and only if the $d$-bit input vector $v$ is light. Since every symmetric function has $\poly(d)$-size formulas~\cite{Khrapchenko1972}, \emph{Light}$(v)$ also has $\poly(d)$-size formulas. Here is the formula:

{\narrower

\begin{framed}
\noindent Take the $\OR$ over all $n^{2-2\alpha(1-\eps)}$ pairs $(i,j) \in [n^{1-\alpha(1-\eps)}]^2$ with $i < j$:\\ 
~~~\indent Take the $\neg\OR$ over all $k=1,\ldots,d$, of the $\AND$ of two items: \\
~~~\indent~~~\indent 1. The $\OR$ over all $O(n^{\alpha(1-\eps)})$ vectors $u$ in group $G_i$ of $(\emph{Light}(u) \wedge u[k])$. \\ 
~~~\indent~~~\indent 2. The $\OR$ over all $O(n^{\alpha(1-\eps)})$ vectors $v$ in group $G_j$ of $(\emph{Light}(v) \wedge v[k])$.
\end{framed}

}

To see that this works, we observe:
\begin{itemize}
\item If there is an orthogonal pair $u,v$ in the instance, then recall that with probability $1-o(1)$, (a) $u$ and $v$ are light, (b) $u$ and $v$ appear in different groups $G_i$ and $G_j$, and (c) there are no other light vectors in $G_i$ and no other light vectors in $G_j$. Thus the inner ORs over the group $G_i$ (and respectively $G_j$) will only output the bits of the vector $u$ (and respectively $v$). Thus the above formula, by guessing the pair $(i,j)$, and checking over all $k=1,\ldots,d$ that $(u[k] \wedge v[k])$ is \emph{not} true, will find that $u,v$ are orthogonal, and output $1$. 
\item If there is no orthogonal pair, then we claim that the formula \emph{always} outputs $0$. Suppose the formula outputs $1$. Then there is some $(i,j)  \in [n^{1-\alpha(1-\eps)}]^2$ such that the inner product of two vectors $V_i$ and $W_j$ is $0$, where $V_i$ is the $\OR$ of \emph{all} light vectors in group $G_i$ and $W_j$ is the $\OR$ of all light vectors in group $G_j$. But for these two vectors to have zero inner product, it must be that \emph{all} pairs of light vectors (one from $G_i$ and one from $G_j$) are orthogonal to each other. Thus there is an orthogonal pair in the instance.
\end{itemize}

Using the $\poly(d)$-size formulas for \emph{Light}, the DeMorgan formula has size 
\begin{align}\label{second-subquad} O(n^{2-2\alpha(1-\eps)} \cdot d \cdot n^{\alpha(1-\eps)} \cdot \poly(d)) \leq O(n^{2-\alpha(1-\eps)} \cdot \poly(d)).\end{align}
Substituting in the value for $\alpha$, the exponent becomes
\[2-\frac{p^3(1-\eps)}{(p+1)^2\log_2(1/(1-p^2))}+ \Theta_p(\eps)\cdot \frac{2}{\log_2(1/(1-p^2))}.\]
Recalling that we are setting $\eps$ to be arbitrarily small (its value only affects the $o(1)$ probability of error), the formula size is \[n^{2-\frac{p^3}{2(p+1)^2\log_2(1/(1-p^2))}+o(1)}.\] Observe that our formula can in fact be made into an $\AC^0$ formula of similar size; this is easy to see except for the $\poly(d)$-size formula for \emph{Light}. But for $d = O(\log n)$, any formula of $\poly(\log n)$-size on $O(\log n)$ bits can be converted into an $\AC^0$ circuit of depth $c/\eps$ and size $2^{(\log n)^{\eps}}$, for some constant $c \geq 1$ and any desired $\eps > 0$.

The final formula is the minimum of the formulas of \eqref{first-subquad} and \eqref{second-subquad}. For every fixed $p \in (0,1]$, we obtain a bound of $n^{2-\eps_p}$ for an $\eps_p > 0$. 
\end{proof}

\section{Conclusion}

It is important to note that the \emph{largest} known lower bound for branching programs computing any explicit function is due to Neciporuk~\cite{N66} from 1966, and is only $\Omega(N^2/\log^2 N)$ for inputs of length $N$. A similar statement holds for Boolean formulas over the full binary basis (see for example~\cite{Jukna12}). Our lower bounds for $\ov{}$ match these bounds up to polylogarithmic factors. Thus it would be a significant breakthrough to generalize our results to other problems believed to require \emph{cubic} time, such as: 

\begin{framed} 
{\sc 3-Orthogonal Vectors} (3-\ov{}) \\
Given: $n$ vectors $v_1,\ldots,v_n \in \{0,1\}^d$\\
Decide: Are there $i,j,k$ such that $\sum_{\ell=1}^d v_i[\ell]\cdot v_j[\ell] \cdot v_k[\ell] = 0$?
\end{framed}

It is known that the Strong Exponential Time Hypothesis also implies that 3-\ov{} requires $n^{3-o(1)}$ for dimensionality $d=\omega(\log n)$~\cite{Williams05,AbboudV14}.

\paragraph{Acknowledgements.} We are very grateful to Ramamohan Paturi for raising the question of whether the OV conjecture is true for AC0 circuits.  

\bibliographystyle{alpha}
\bibliography{cc-papers,mypubs}

\appendix

\section{Chernoff Bound}\label{chernoff}

We use the following standard tail bound:

\begin{theorem} \label{chernoff-bd} Let $p \in (0,1)$ and let $X_1,\ldots,X_d \in \{0,1\}$ be independent random variables, such that for all $i$ we have $\Pr[X_i=1] = p$. Then for all $\delta \in (0,1)$,
\[\Pr\left[\sum_i X_i > (1-\delta)pd\right] \leq e^{-\delta^2 pd/2}.\]
\end{theorem}

\section{DeMorgan Formulas into Branching Programs}\label{deMorgan-bp}

Here we describe at a high level how to convert a DeMorgan formula (over $\AND$, $\OR$, $\NOT$) of size $s$ into a branching program of size $O(s)$. 

Our branching program will perform an in-order traversal of the DeMorgan formula, maintaining a counter (from $1$ to $s$) of the current node being visited in the formula. The branching program begins at the root (output) of the formula. If the current node is a leaf, its value $b$ is returned to the parent node. If the current node is not a leaf, the branching program recursively evaluates its left child (storing no memory about the current node). 

The left child returns a value $b$. If the current node is an $\AND$ and $b=0$, or the current node is an $\OR$ and $b=1$, the branching program propagates the bit $b$ up the tree (moving up to the parent). If the current node is a $\NOT$ then the branching program moves to the parent with the value $\neg b$. 

If none of these cases hold, then the branching program erases the value $b$, and recursively evaluates the right child, which returns a value $b$. This value is simply propagated up the tree (note the fact that we visited the right child means that we know what the left child's value was). 

Observe that we only hold the current node of the formula in memory, as well as $O(1)$ extra bits. 

\end{document}